\documentclass[10pt,journal]{IEEEtran}

\usepackage{graphicx}
\usepackage{subcaption}
\usepackage{amsmath}

\usepackage{amsfonts}
\usepackage{float}
\usepackage{mathtools}
\usepackage{tikz}

\usepackage{algorithm}
\usepackage[noend]{algpseudocode}

\newtheorem{lemma}{Lemma}

\newcommand{\hvec}[1]{\ensuremath{\Hat{\vec{#1}}}}
\newcommand{\bvec}[1]{\ensuremath{\Breve{\vec{#1}}}}
\renewcommand{\vec}[1]{\ensuremath{\vec{#1}}}

\newcommand{\norm}[1]{\ensuremath{\| #1 \|}}
\newcommand{\mc}[1]{\ensuremath{\mathcal{#1}}}

\newcommand{\of}[1]{^{(#1)}}

\newcommand{\floor}[1]{\lfloor #1 \rfloor}

\DeclareMathOperator*{\argmax}{arg\,max}

\renewcommand{\eqref}[1]{(\ref{eq:#1})}

\newcommand{\figref}[1]{Fig.~\ref{fig:#1}}

\newcommand{\secref}[1]{Section~\ref{sec:#1}}

\newcommand{\lemref}[1]{Lemma~\ref{lem:#1}}

\renewcommand{\vec}[1]{\ensuremath{\mathbf{#1}}}

\begin{document}
\title{Communications using Sparse Signals}

\author{
\IEEEauthorblockN{Madhusudan Kumar Sinha\IEEEauthorrefmark{1}, Arun Pachai Kannu\IEEEauthorrefmark{2}} \\
\IEEEauthorblockA{Department of Electrical Engineering\\
Indian Institute of Technology, Madras\\
Chennai, Tamil Nadu - 600036\\
Email: \IEEEauthorrefmark{1}ee16d028@smail.iitm.ac.in, \IEEEauthorrefmark{2}arunpachai@ee.iitm.ac.in }
}

\maketitle

\begin{abstract}

Inspired by compressive sensing principles, we propose novel  error control coding techniques for communication systems. The information bits are encoded in 
the support and the non-zero entries of a sparse signal. By selecting a dictionary matrix with suitable dimensions, the codeword for transmission is obtained by multiplying the dictionary matrix with the sparse signal. Specifically, the codewords  are obtained from the sparse linear combinations of the columns of the 
dictionary matrix. At the decoder, we employ variations of greedy sparse signal recovery algorithms. 
Using Gold code sequences and mutually unbiased bases from quantum information theory as dictionary matrices, we study the block error rate (BLER) performance of the proposed 
scheme in the AWGN channel. Our results show that the proposed scheme has a comparable and competitive performance with respect to the several widely used linear 
codes, for very small to moderate block lengths. In addition, our coding scheme extends straightforwardly to multi-user scenarios such as, multiple access channel, 
broadcast channel and interference channel. In these multi-user channels, if the users are grouped such that they have similar channel gains and noise levels, 
the overall BLER performance of our proposed scheme will coincide with an equivalent single user scenario. 

\end{abstract}



%
\IEEEpeerreviewmaketitle

\section{Introduction}

Shannon's seminal paper on information theory established the existence of information encoding and decoding techniques which guarantee almost error free
communications across noisy channels \cite{shannon1948mathematical}. Extensive work has been carried out to develop such efficient error control 
coding techniques for additive white Gaussian noise (AWGN) channels \cite{lin2001error}. 
Linear codes such as convolutional codes, turbo codes and LDPC codes 
are widely used in various communication systems today. Maximum likelihood decoding in AWGN channels boils down to finding 
the codeword in the codebook which is closest to the received signal \cite{madhow2008fundamentals}. Hence the distance properties of the codewords in the 
codebook play an important role in the error performance of the coding scheme. 

Code division multiple access (CDMA) is a communication technique developed for multi-user wireless systems \cite{verdu1998multiuser}. 
Each user is given a specific 
code or a sequence from a large set of sequences. Users multiply their information bearing symbols with the sequences assigned to them in a 
spreading operation.  The received signal is superposition of signals from all the users. When the receiver does despreading operation on the received signal using the sequence of a given user (which is equivalent to finding inner product), the interference from 
other users are suppressed to a large extent, if the correlation between the sequences is small. There is extensive literature on finding a large set of 
codes/sequences with good correlation properties. For instance, Gold codes \cite{Goldcode} and Zadoff-Chu sequences \cite{frank1963polyphase,chu1972polyphase} are well-known for 
their correlation properties and widely used in wireless systems. In addition, quantum information theory also provides ways to construct a large set of 
sequences with small correlation among them. Such constructions are referred as mutually unbiased bases (MUB) \cite{wootters1989optimal}
and symmetric, informationally complete, positive operator valued measure (SIC-POVM) 
\cite{renes2004symmetric}. SIC-POVM is closely related to construction 
of equi-angular lines.

Compressive sensing techniques address the problem of recovering a sparse signal from an under-determined system of noisy linear measurements \cite{eldar2012compressed}.
Suppose $\vec{x}$ is an $L$-dimensional signal/vector with sparsity level $K$ such that only $K$ entries in $\vec{x}$ are non-zero. Using the 
sensing matrix $\vec{A}$ having size $N \times L$ with $K < N < L$, and the noisy linear measurement vector $\vec{y} = \vec{A} \vec{x} + \vec{v}$, the goal is to recover the support and non-zero entries of the $K$-sparse signal $\vec{x}$ from $\vec{y}$. There is extensive literature on the sparse signal recovery algorithms, 
such as greedy matching pursuit based algorithms  \cite{mallat1993matching,pati1993orthogonal},
 convex programming
based algorithms \cite{chen2001atomic,candes2007dantzig}, approximate message passing algorithms \cite{maleki2011approximate} and their deep networks based implementations \cite{metzler2017learned}. 
The performance of the sparse signal recovery algorithms depend on the mutual coherence parameter of the sensing matrix \cite{candes2007sparsity}, which is directly related
to the correlation among its columns. Smaller the correlation among the columns of $\vec{A}$, better is the sparse signal recovery performance.

There are inherent connections between error control codes, CDMA sequences and compressive sensing, as they all require small correlation among codewords/spreading-sequences/sensing-matrix-columns. There is a vast literature on connecting compressive sensing concepts with 
communication techniques. We have highlighted some of these works here. 
In \cite{candes2005error}, error control coding techniques have been developed for the case of sparse noise vector 
(which models impulse noise environments). Codes/sequences with small correlation have been used with sparse signal recovery techniques 
for a massive random access application in \cite{jain2020algorithms}. Techniques to compute the sparse Fourier transform
using the parity check matrix of an LDPC code as the sensing matrix along with peeling decoder have been developed in  \cite{pawar2013pulse}. Spatial modulation 
in multi-antenna systems encodes the information partially by activating a subset of antennas chosen from a large set \cite{mesleh2008spatial}. Index modulation in OFDM systems encode information partially in choosing a subset of sub-carriers among the available sub-carriers to send the data symbols \cite{abu2009subcarrier}.

In this paper, we develop a new error control coding scheme using sequences with low correlation and compressive sensing concepts. 
We develop a subblock sparse coding (SSC) scheme where information bits are \emph{non-linearly}
encoded in a $K$-sparse signal $\vec{x}$ of length $L$, with the non-zero entries chosen from an 
$M$-ary constellation. The SSC scheme carries roughly $K \log \frac{L}{K}$ bits in the 
support of $\vec{x}$ and $K \log M$ bits in the non-zero entries. The codeword $\vec{s}$ of length $N$ to be transmitted 
across the channel is obtained by multiplying $\vec{x}$ with a dictionary matrix $\vec{A}$ (of size $N \times L$). The columns 
of the dictionary matrix are chosen from the set of low correlation sequences from CDMA or quantum information theory. 
From the noisy observation $\vec{y} = \vec{s} + \vec{v} = \vec{A}\vec{x} + \vec{v}$,  we recover the sparse signal $\vec{x}$
(and hence the information bits) using a novel greedy match and decode (MAD) algorithm and its variations. With Gold codes from CDMA and mutually
unbiased bases from quantum information theory as the dictionary matrices,  
the codeword error rate performance of the proposed SSC encoding scheme with  MAD decoding  in AWGN channel is comparable and  
competitive with the widely used binary linear codes for very small ($N=8$) to moderate codeword lengths ($N=128$).
In addition, the proposed error control coding scheme extends easily to multi-user 
channels, such as multiple-access channel (MAC), broadcast channel (BC) and interference channels (IC). 
SSC scheme transmitting $B$ bits to a single user using a $K$-sparse signal can be easily modified to transmit a \emph{total} of $B$ bits to 
$P$ users with $P \leq K$, in a MAC or BC or IC. In addition, the overall error performance in the multi-user channel will be same
as that of an equivalent single user scenario. 
%

While index modulation in OFDM and spatial modulation techniques encode information partially in choosing a subset from 
available sub-carriers/antennas, these techniques require an underlying error control coding scheme to ensure small probability of error. On the other hand, our SSC encoding and MAD decoding is a new error control scheme by itself. 

Our work has connections to the work on non-linear
codes by Kerdock and Preparata \cite{kerdock1972class,preparata1968class}. Kerdock's work gave constructions for a large set of sequences with good 
distance/correlation properties and used them directly as codewords. Spherical codes \cite{delsarte1991spherical} also aim at developing a large set of codewords with good distance properties. In our work, we generate codewords using \emph{sparse} linear combinations of sequences with good correlation properties. Because of these linear combinations, the number of codewords in our SSC scheme is larger when compared to that of Kerdock codes and spherical codes of similar lengths. This increase in the number of codewords has two benefits, increase in the data rate and decrease in the overall energy per bit.  

Our SSC encoding scheme has a non-linear component (mapping from bits to sparse signal) and a linear component (mapping from
sparse signal to codeword). Due to the linear part, the recovery of sparse signal from the observation can be accomplished using simple decoding techniques. At the same time, the non-linear component in the encoding procedure enables direct extensions of the scheme 
to the multi-user channels.

The paper is organized as follows. In \secref{enc}, we describe the proposed sparse signal based encoding and decoding techniques. In \secref{dic}, we 
present the details of constructing dictionary matrices using Gold codes and complex mutually unbiased bases. In \secref{sims}, we present simulation 
studies on the error performance of the proposed schemes in AWGN channel and compare with some of the existing error control codes. 
In \secref{mus}, we discuss the details of extending the proposed coding techniques to multi-user scenarios. In \secref{con}, we give concluding remarks and directions for future work. 

\section{Encoding and Decoding Schemes} \label{sec:enc}

\subsection{Sparse Coding} \label{sec:sc}

Consider a dictionary matrix $\vec{A}$ of size $N \times L$, with $L \geq N$. The codewords for messages are obtained using \emph{sparse} linear combinations 
of columns of the matrix $\vec{A}$. We discuss the details of sparse encoding procedure below. 

\subsubsection{Encoding Procedure}
Fix the \emph{sparsity} level as $K$ with $K \leq N$. Choose a subset $\mc{S} \subset \{1,\cdots,L\}$ of size $\left|\mc{S}\right| = K$. Let us denote $\mc{S} = \{
\alpha_1,\cdots,\alpha_K\}$ with $1 \leq \alpha_k \leq L$. Let $\mc{Q} = \{\beta_1,\cdots,\beta_K\}$ be an ordered set of $K$ symbols chosen (allowing repetitions) from a $M-$ary constellation, with alphabet set  $\mc{B} = \{b_1,\cdots,b_M\}$. 
Denoting $i^{th}$ column of $\vec{A}$ by $\vec{a}_i$, a codeword of length $N$ is obtained as 
\begin{equation}
\vec{s} = \sum_{k=1}^K \beta_k \vec{a}_{\alpha_k}. \label{eq:cw1}
\end{equation}
Consider the sparse vector $\vec{x}$ of length $L$, with its $i^{th}$ entry given as
\begin{equation}
x_i = \left\{ \begin{array}{ll} \beta_k &  \text{if}~ i = \alpha_k  \\ 0 & \text{if}~ i \notin \mc{S} \end{array} \right. \label{eq:ss}
\end{equation}
Now, the codeword in \eqref{cw1} can be represented as
\begin{equation}
\vec{s} = \vec{A}\vec{x}. \label{eq:cw2}
\end{equation}
The information is encoded in the support set $\mc{S}$ of the sparse signal $\vec{x}$ and its non-zero entries given by the set $\mc{Q}$.  Let $\mc{C}$ denote the set of all possible codewords of the form \eqref{cw2}, with a fixed sparsity level $K$. Total number of codewords we can generate is
$|\mc{C}|=M^K \times \binom{L}{K}$. The total number of bits that can be transmitted in a block of $N$ channel uses is
\begin{equation}
N_b = K \floor{\log M } + \floor{ \log \binom{L}{K}}. \label{eq:nb1}
\end{equation}
In this paper, base of $\log(\cdot)$ is $2$, unless specified explicitly otherwise.
We define the \emph{code rate} of the encoding scheme in units of bits per real dimension (bpd) as the number of bits transmitted per  
real dimension utilized. If $\vec{A}$ is a real matrix and the modulation symbols in $\mc{Q}$ are chosen from a real constellation (such as PAM, BPSK), 
the code rate in bpd is $\frac{N_b}{N}$.  On the other hand, if $\vec{A}$ is a complex matrix and/or the constellation symbols are complex, 
the code rate in bpd is $\frac{N_b}{2N}$. In our encoding process, we also allow the special case of $M=1$, for which $\beta_k = +1, \forall k$. 
The average energy per bit of our sparse encoding scheme can be given as $\displaystyle E_b = \frac{1}{N_b} \frac{1}{|\mc{C}|} \sum_{\vec{s} \in \mc{C}} \|\vec{s}\|^2$. 
This proposed coding scheme is a non-linear code and can be considered as a generalization of orthogonal FSK. Note that, with the special case of $\vec{A}$ being a DFT matrix and setting $K=1$ and $M=1$, our encoding process leads to an orthogonal FSK scheme. An interesting analogy is that the columns of the dictionary matrix can be compared
to the words in a dictionary of a language. With this analogy, codewords are equivalent to sentences in a language, as they are obtained by using sparse combinations of words and different codewords/sentences convey different messages. 

\subsubsection{Decoding Procedure}
The received signal is modeled as
\begin{eqnarray}
\vec{y} &=& \vec{s} + \vec{v}, \nonumber \\ 
&=& \vec{A}\vec{x} + \vec{v}, \label{eq:obs1}
\end{eqnarray}
where $\vec{v}$ is additive noise. Information bits can be retrieved by recovering the sparse signal $\vec{x}$ from the observation $\vec{y}$. Sparse signal recovery 
can be done using greedy techniques \cite{mallat1993matching,cai2011orthogonal,tropp2004greed} or convex programming based techniques \cite{chen2001atomic}. In this paper, we consider a simple greedy algorithm which 
we refer as match and decode (MAD) described in Algorithm \ref{mad}. MAD algorithm takes the dictionary matrix $\vec{A}$, the observation $\vec{y}$, sparsity level $K$ as inputs and produce an estimate 
$\hvec{x}\of{K}$ of the sparse signal $\vec{x}$. It is ensured that the estimate 
$\hvec{x}\of{K}$ (of size $L$) has exactly $K$ non-zero entries from the constellation set $\mc{B}$.
Any sparse signal $\hvec{x}$ (of size $L$) with at most $K$ non-zero entries from set $\mc{B}$ can also be given as partial information to the MAD algorithm. 
If no partial information is available, $\hvec{x}$ is set as $\vec{0}$.   

\begin{algorithm}
\caption{Match and Decode Algorithm}\label{mad}
\begin{algorithmic}[1]

\State \textbf{Input:} Get the observarion $\vec{y}$, dictionary matrix $\vec{A}$, sparsity level $K$ and any partially recovered sparse signal  $\hvec{x}$. 

\State \textbf{Initialize:}  
Let the ordered sets $\hat{\mc{S}}$ and $\hat{\mc{Q}}$ denote the the support and the corresponding non-zero entries of partial information 
vector $\hvec{x}$. If $\hvec{x}=\vec{0}$, $\hat{\mc{S}}$ and $\hat{\mc{Q}}$ are empty sets. 
Initialize the iteration counter $t=|\hat{\mc{S}}|$, the residual $\vec{r}\of{t} = \vec{y} - \vec{A} \hvec{x}$ and the estimate $\hvec{x}\of{t} = \hvec{x}$.

\State \textbf{Match:} Correlate the residual with the columns of the dictionary matrix and the constellation symbols as given below. 
\begin{align}
c_i  &= \langle \vec{r}\of{t} , \vec{a}_i \rangle, ~~ i \in \{1,\cdots,L\} ~\& ~i \notin \hat{\mc{S}} \label{eq:metri0} \\
p_{i,m} &= \mathfrak{Real}\{c_i b_m^*\} - \frac{|b_m|^2}{2}, ~~ b_m \in \mc{B} \label{eq:metri}
\end{align}

\State \textbf{Decode:} Detect the active column and the corresponding non-zero entry. 
$(\hat{i},\hat{m}) = \arg\max_{\substack{i \notin \hat{\mc{S}} \\ 1 \leq m \leq M}} p_{i,m}$

\State \textbf{Update:} Update the recovered sparse signal information,  $\hat{\mc{S}} = \hat{\mc{S}} \cup \hat{i}$, $\hat{\mc{Q}} = \hat{\mc{Q}} \cup b_{\hat{m}}$, $\hvec{x}\of{t+1} = \hvec{x}\of{t} + b_{\hat{m}} \vec{e}_{\hat{i}}$. (Here $\vec{e}_n$ denotes $n^{th}$ standard basis). Update the residual  
\begin{align}
\vec{r}\of{t+1} &=   \vec{r}\of{t} -   b_{\hat{m}} \vec{a}_{\hat{i}}, \label{eq:resi} 
\end{align}
and increment the counter $t = t+1$.

\State \textbf{Stopping condition:} If $t < K$, repeat the above steps Match, Decode and Update. Else go to Step Ouptut.

\State \textbf{Output:} Recovered sparse signal is $\hvec{x}\of{K}$ and the recovered codeword $\hvec{s} = \vec{A} \hvec{x}\of{K}$. 

\end{algorithmic}
\end{algorithm}

We  note that, the correlation of residual with the columns of dictionary matrix in \eqref{metri0} needs to be computed only for one iteration. 
For the subsequent iterations, from \eqref{resi}, 
we have the recursion, $\langle \vec{r}\of{t+1} ,\vec{a}_i \rangle  = \langle \vec{r}\of{t},\vec{a}_i\> \rangle - 
 b_{\hat{m}} \langle \vec{a}_{\hat{i}} ,\vec{a}_i \rangle$, where
$b_{\hat{m}}$ and $\hat{i}$ denote the symbol and the active column detected in the previous iteration. We can store the symmetric gram matrix $\vec{A}^* \vec{A}$, 
to get the values of $\langle \vec{a}_{\hat{i}} ,\vec{a}_i \rangle$ needed in the recursion. 

Intuitively, the first iteration of MAD algorithm is the most error prone, since it faces the \emph{interference} from all the undetected columns. To improve on MAD performance,
we consider a variation, referred as parallel MAD. In the first iteration, we choose $T$ candidates for the active column, by taking the top $T$ metrics \eqref{metri}, and perform MAD decoding for each of these $T$ candidates, resulting in $T$ different estimates for the sparse signal. Among these $T$ estimates, we select the one with the smallest Euclidean distance to the observation, inspired by the decoder for white Gaussian noise. The mathematical details are described in Algorithm \ref{pmad} for completeness.

\begin{algorithm}
\caption{Parallel Match and Decode Algorithm} \label{pmad}
\begin{algorithmic}[1]

\State Given the dictionary matrix $\vec{A}$ and the observation vector $\vec{y}$, 
compute $c_i = \langle \vec{y} , \vec{a}_i \rangle,~i=1,\cdots,L$ and $p_{i,m} =  
\mathfrak{Real}\{c_i b_m^*\} - \frac{|b_m|^2}{2}, ~ b_m \in \mc{B}$.

\State Initialize parallel path index $n=1$; Initialize $\mc{D} = \emptyset$.  

\State Choose a candidate for active column and the corresponding non-zero entry: $(\hat{i}_n,\hat{m}_n) = \arg\max_{(i \notin \mc{D} ,m)} p_{i,m}$.

\State  Run MAD algorithm with inputs $(\vec{A},\vec{y},K)$ and prior information on sparse signal $\hvec{x} = b_{\hat{m}_n} e_{\hat{i}_n}$. 
Denote the recovered sparse signal output of MAD as $\hvec{x}_n$. 

\State Update $\mc{D} = \mc{D} \cup \hat{i}_n$ and $n=n+1$; If $n \leq T$, go back to Step 3.

\State Final output $\bvec{x} = \arg\min_{\hvec{x}_n; 1 \leq n \leq T} \| \vec{y}-\vec{A}\hvec{x}_n\|$.

\end{algorithmic}
\end{algorithm}

Main computationally intensive step in MAD (and parallel MAD) is computing the correlation between the observation and the columns of dictionary matrix in \eqref{metri0}, which amounts to computing the matrix multiplication $\vec{A}^* \vec{y}$. Depending on the choice of $\vec{A}$, efficient matrix multiplication techniques may be developed.

\subsubsection{Performance Guarantees}
With $\vec{s}$ being the transmit codeword and $\hvec{s}$ being the codeword recovered by the decoding algorithm, the event $\{\vec{s} \neq \hvec{s}\}$ results in a 
block error (at least one of the bits in the block is decoded in error). The mutual coherence of the dictionary matrix $\vec{A}$ defined below,
\begin{equation}
\mu = \max_{p \neq q}  \frac{\langle \vec{a}_p, \vec{a}_q \rangle}{\norm{\vec{a}_p} \norm{\vec{a}_q}},
\end{equation}
plays an important role in the recovery performance of MAD algorithm. 

\begin{lemma} \label{lem:rec}
In the absence of noise ($\vec{v} = \vec{0}$), and for the case of $M=1$, MAD algorithm recover the codeword perfectly, if $K < \frac{1}{2}(\mu^{-1} + 1)$.
\end{lemma} 
\begin{proof}
This result has already been established for orthogonal matching pursuit (OMP) algorithm in \cite{tropp2004greed}. We note that
MAD differs from  OMP in certain aspects. For the $M=1$ case, MAD algorithm sets the non-zero entries as unity. This is better than OMP, which uses least squares estimates for the non-zero entries in each iteration. In computing the residual, MAD subtracts out the detected columns from the observation. This is better than OMP, 
which projects the observation onto the orthogonal complement of the 
detected columns, possibly reducing the signal component from yet-to-be-detected columns. 
Hence, MAD recovery will be at least as good as OMP recovery, when $\vec{v}=\vec{0}$ and $M=1$. 
\end{proof}

The exact support recovery performance of OMP in the presence of bounded noise and Gaussian noise are characterized in \cite{cai2011orthogonal}. 
The same results hold true for MAD algorithm as well, when $M=1$. When $M>1$, the error event can also happen due to incorrect
decoding of the modulation symbols in the constellation. Characterizing that error event will depend on the exact constellation shape and 
this analysis can be carried out in a future work.


\subsection{Subblock Sparse Coding} \label{sec:ssc}
A drawback of the sparse coding scheme in \secref{sc} is that a large look-up table is needed to map the information bits to the sparse signals. 
In order to eliminate the look-up table, we propose a subblock sparse coding (SSC) scheme described below. In this scheme, 
we partition the dictionary matrix $\vec{A}$ into $K$ subblocks such that $\vec{A} = \left[\vec{A}_1 \cdots \vec{A}_K \right]$ with $k^{th}$ subblock $\vec{A}_k$ 
having a size of $N \times L_k$ and $\displaystyle \sum_{k=1}^K L_k = L$. This partitioning, in general, can be done in an arbitrary manner. Later in this section, we present
a  partitioning technique with the aim of maximizing the number of information bits encoded by the scheme. 
SSC technique transmits $N_b$ number of bits (in $N$ channel uses), where 
\begin{equation}
N_b = K \floor{ \log M } + \sum_{k=1}^K \floor{\log L_k}.  \label{eq:nb}
\end{equation}
From a bit stream of length $N_b$, we take the first $K \floor{ \log M }$ bits and obtain $K$ modulation symbols $\mc{Q} = \{\beta_1,\cdots,\beta_K\}$ from an $M$-ary constellation. Now, we segment the remaining bit stream into $K$ strings with $k^{th}$ string $\vec{b}_k$ having a length of $\floor{\log L_k}$. Let $N_k$ denote the unsigned integer corresponding to the bit string $\vec{b}_k$. Now, from each subblock $\vec{A}_k$, we select the column indexed by the number $N_k+1$, for $1\leq k \leq K$. 
Note that, based on this procedure, the column indices chosen from the original matrix $\vec{A}$ is given by 
$\alpha_k = N_k+1 +  \sum_{i = 1}^{k-1} L_i $, $1 \leq k \leq K$. With support set $\mc{S} = \{ \alpha_1,\cdots,\alpha_K\}$ and the modulation symbols in $\mc{Q}$, 
the codeword $\vec{s}$ is obtained as given in \eqref{cw1}. By this  subblock encoding procedure, we avoid the exhaustive 
look-up table needed to map the bits to the codewords. However, the dictionary matrix $\vec{A}$ needs to be stored. 
In addition, the number of bits transmitted in a block for SSC scheme will be less than that of sparse coding scheme discussed in 
\secref{sc}. 
For the SSC scheme, the MAD algorithm can be modified to discard the subblock corresponding to each detected column from the subsequent iterations. 

With the aim of maximizing $\sum_{k=1}^K \floor{\log L_k}$, we describe a procedure for truncating and partitioning a matrix with 
$\bar{L}$ columns into $K$ subblocks (where $K \leq \frac{\bar{L}}{2}$) with $k^{th}$ subblock having $L_k$ columns, 
such that each $L_k$ is a power of $2$,  and $\sum_k L_k = L \leq \bar{L}$. 
The algorithm has $K$ iterations and the $n^{th}$ iteration with $n \in \{1,\cdots,K\}$, involves partitioning the dictionary matrix into $n$ subblocks, with lengths given by $\{L_1\of{n},\cdots,L_n\of{n}\}$. 
In the first iteration $n=1$, we set  $L_1\of{1} =  2^{\floor{\log \bar{L}}}$. In the iteration $n$, we make use of the subblock lengths obtained in the previous iteration, 
arranged in the ascending order such that 
$L_1\of{n-1} \leq L_2\of{n-1} \cdots \leq L_{n-1}\of{n-1}$. We compute the number of \emph{remaining} columns from the iteration $n-1$ as $r\of{n-1} = \bar{L} - \sum_{k=1}^{n-1} L_k\of{n-1}$. 
If $r\of{n-1} \geq \frac{L_{n-1}\of{n-1}}{2}$, then we create a \emph{new} subblock of length $L_n\of{n} = 2^{\floor{ \log r\of{n-1}}}$ and retain all the previous subblocks so that $L_1\of{n} = L_1\of{n-1}, 
\cdots, L_{n-1}\of{n} = L_{n-1}\of{n-1}$. On the other hand, if $r\of{n-1} < \frac{L_{n-1}\of{n-1}}{2}$, we split the largest subblock from the previous iteration into two equal parts so that 
$L_{n-1}\of{n} = \frac{L_{n-1}\of{n-1}}{2}, L_{n}\of{n} = \frac{L_{n-1}\of{n-1}}{2}$ and the remaining subblocks are retained so that 
$L_1\of{n} = L_1\of{n-1}, \cdots, L_{n-2}\of{n} = L_{n-2}\of{n-1}$. We sort the lengths $\{ L_1\of{n},\cdots,L_n\of{n} \}$ in ascending order, for use in the subsequent iteration. 
At the end of iteration $K$, we get the required lengths for $K$ subblocks. Using mathematical induction, we can argue that the above procedure is optimal in
maximizing $\sum_{k=1}^K \floor{\log L_k}$.




\section{Dictionary Matrix Construction} \label{sec:dic}

The choice of the dictionary matrix $\vec{A}$ plays a vital role in the block error performance. It is desirable that the dictionary matrix has a large number of columns (as the number of information bits increases with $L$, for fixed $N$ and $K$) with small correlation among the columns (for good sparse signal recovery performance). In this paper, we consider  dictionary matrix constructions using Gold code sequences from CDMA literature and mutually unbiased bases from quantum information theory. 

\subsection{Gold Codes}
Gold codes are binary sequences with alphabets $\{\pm 1\}$. Considering lengths of the form $N = 2^n-1$, where $n$ is any positive integer, there are  $2^n+1$ Gold sequences. By considering all the circular shifts of these sequences, we get $2^{2n}-1$ sequences. When dictionary matrix columns are  constructed with these $2^{2n}-1$ sequences normalized to unit norm, the resulting mutual coherence $\mu$ is given by \cite{Goldcode}, 
\begin{equation}
    \mu=\begin{cases} 
      \frac{2^\frac{n+1}{2}}{N}, & n \text{ is odd,} \\
      \frac{2^\frac{n+2}{2}}{N}, & n \text{ is even.}
   \end{cases} \label{eq:mugold}
\end{equation}
We note that odd value of $n$ leads to relatively smaller mutual coherence.  We can add any column of the identity matrix to the Gold code dictionary matrix, to get a total of $L=2^{2n}$ columns (which is a power of 2), with the mutual coherence same as \eqref{mugold}. 
Storing such a dictionary matrix will require $N(N+1)^2$ bits. For $N=127$, this Gold code dictionary matrix size is approximately $2.1$ MB.

\subsection{Mutually Unbiased Bases}


Two orthonormal bases $\mc{U}_1$ and $\mc{U}_2$ of the $N$-dimensional inner product space $\mathbb{C}^N$ are called mutually unbiased if and only if $|\langle\vec{x},\vec{y}\rangle|=\frac{1}{\sqrt{N}}$ for any $\vec{x}\in \mc{U}_1$ and $\vec{y}\in \mc{U}_2$. A set of $m$ orthonormal bases of $\mathbb{C}^N$ is called mutually unbiased if all bases in the set are pairwise mutually unbiased.
Let $Q(N)$ denote the maximum number of orthonormal bases of $\mathbb{C}^N$, which are pairwise mutually unbiased. In \cite{wootters1989optimal}, it has been shown  that $Q(N) \leq N$ (excluding the standard basis), with equality if $N$ is a prime power. Explicit constructions are also given in \cite{wootters1989optimal} for 
getting $N$ MUB in $N$-dimensional complex vector space $\mathbb{C}^N$, if $N$ is a prime power.





If $N$ is a power of 2, the $N$ MUB unitary matrices $\{\vec{U}_1,\cdots,\vec{U}_N\}$ have the following properties. 
\begin{itemize}
\item The entries in all the $N$ unitary matrices belong to the set $\{\frac{+1}{\sqrt{N}},\frac{-1}{\sqrt{N}},\frac{+i}{\sqrt{N}},\frac{-i}{\sqrt{N}}\}$. This 
follows from the construction of MUB given in \cite{wootters1989optimal}. Storing all these $N$ unitary matrices will require $2 N^3$ bits. For $N=128$, this storage requirement is approximately $4.2$ MB.
\item For $N$ up to 256, we find that the inner products  $\langle\vec{x},\vec{y}\rangle \in \{\frac{+1}{\sqrt{N}},\frac{-1}{\sqrt{N}},\frac{+i}{\sqrt{N}},\frac{-i}{\sqrt{N}}\}$ 
when $\vec{x}\in \vec{U}_i$ and $\vec{y}\in \vec{U}_j$ for $i \neq j$. We conjecture that this property holds true when $N$ is any power of 2. 
\end{itemize}

We construct dictionary matrix using $N$ MUB as $\vec{A} = [\vec{U}_1 \cdots \vec{U}_N]$. In this case, $L = N^2$ and the corresponding mutual coherence $\mu$ is 
$\frac{1}{\sqrt{N}}$. In addition, when $N$ is a power of $2$, we can always split $\vec{A}$ into $K$ subblocks with size of each subblock $L_k$ being a power of $2$ and each 
$L_k \geq \frac{N^2}{2K}$. 
When the entries in $\vec{A}$ belong to $\{+1,-1,+j,-j\}$ except for a common scaling factor of $\frac{1}{\sqrt{N}}$, the inner products $\langle \vec{a}_i,\vec{y} \rangle$
needed in the MAD algorithm involve only additions (no multiplications).

\section{Simulation Results} \label{sec:sims}

We study the performance of the block error rate (BLER), also referred as codeword error rate, for the proposed encoding and decoding schemes in additive 
white Gaussian noise channels. For the complex MUB dictionary matrix, the non-zero entries of the sparse signal are chosen from QPSK constellation. For real 
Gold code dictionary matrix, we consider BPSK constellation. When the non-zero entries in the $K$-sparse signal $\vec{x}$ are uncorrelated, it easily follows that, the 
expected energy of the codeword $\vec{s}=\vec{A}\vec{x}$ is $E_s = K$, when the columns of dictionary matrix are of unit norm. Energy per bit $E_b$ is obtained 
by dividing $E_s$ by the total number of bits conveyed by the sparse signal $\vec{x}$. With $\frac{N_0}{2}$ denoting the variance of the Gaussian noise per real 
dimension, we plot the BLER versus $E_b/N_0$ of the proposed schemes.     





\subsection{Study of Proposed Schemes}

In this section, we study various combinations of the proposed encoding and decoding schemes to understand their impact on the performance. In this subsection, we restrict our attention to complex MUB dictionary matrix with $N=64$.

In Fig. \ref{fig:ex_vs_sb}, we compare the performance of sparse coding (SC) scheme and the subblock sparse coding scheme, with MAD decoding. 
Since SC scheme allows all the 
$\binom{L}{K}$  possibilities for the support set, the number of bits $N_b$ per block of SC scheme will be larger than SSC scheme. Specifically, with complex MUB 
dictionary matrix of length $N=64$, for sparsity levels $K=1,3,5$, the values of $N_b$ for SC scheme are, respectively, $14,39,63$ and the corresponding values for SSC scheme are $14,37,58$. Hence, when 
$E_b/N_0$ of the two schemes are same, SSC scheme has a larger value for noise variance parameter $N_0$. On the other hand, the MAD decoder faces a larger search space to detect each active column of SC scheme, while in SSC scheme, once an active column is detected, the entire subblock can be removed from subsequent iterations. These two effects counteract each other and we find that the BLER versus $E_b/N_0$ performance of both SC and SSC schemes are quite close to each other. But SC scheme has higher code rates at the expense of extensive look up table needed to map the information bits to the (support of the) sparse signal. 


\begin{figure}
    \centering
    \includegraphics{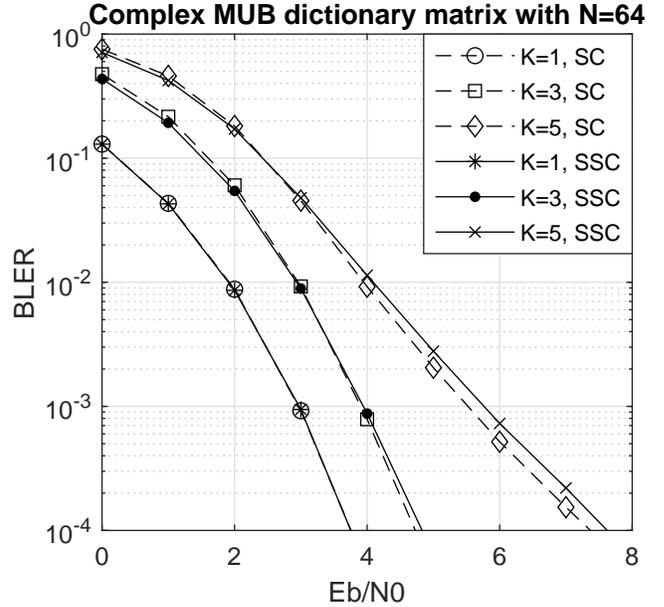}
    \caption{Impact of the Encoding Scheme.}
    \label{fig:ex_vs_sb}
\end{figure}

In Fig. \ref{fig:MAD}, we compare the performance of sparse coding scheme with MAD decoder and OMP algorithm \cite{tropp2004greed}. In the OMP algorithm, in each iteration, after an active column is identified based on the magnitude of correlation with the residual,  the least squares estimates of the non-zero entries are quantized to the nearest constellation points. On the other hand, MAD decoder utilizes the finite alphabet size of the non-zero entries by 
jointly decoding the active column and the corresponding constellation point. In addition, the OMP approach of projecting the residuals onto the orthogonal complement of
the detected columns in each iteration leads to reduction of signal components from the yet-to-be detected active columns. On the other hand, MAD simply subtracts out the detected columns without affecting the yet-to-be detected active columns. Due to these reasons, the proposed MAD decoder provides better performance, when compared to the OMP algorithm.

\begin{figure}
    \centering
    \includegraphics{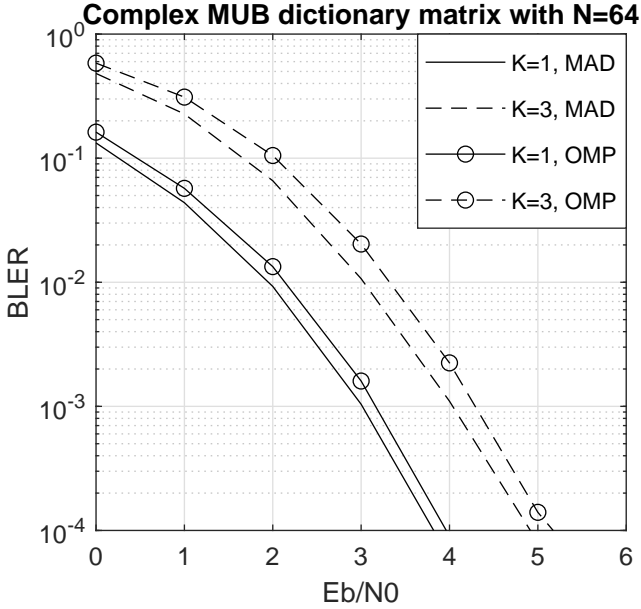}
    \caption{Impact of the Decoding Algorithm.}
    \label{fig:MAD}
\end{figure}

In Fig. \ref{fig:MAD_vs_pMAD}, we compare the performance of MAD and parallel MAD decoding for SSC scheme. 
 The probability of selecting a wrong column in the first iteration of MAD decoding increases with sparsity level $K$ due to interference from 
$K-1$ remaining columns, especially when $K$ is close to $\frac{1}{2\mu}$. Parallel MAD overcomes this problem by selecting $T$ candidates for the active column
in the first iteration and subsequently running $T$ parallel MAD decoders. In all our simulations, we set the value of $T$ equal to the sparsity level $K$. The results show significant gains of parallel MAD over MAD when $K=5$ 
with complex MUB dictionary matrix of size $N=64$.



\begin{figure}
    \centering
    \includegraphics{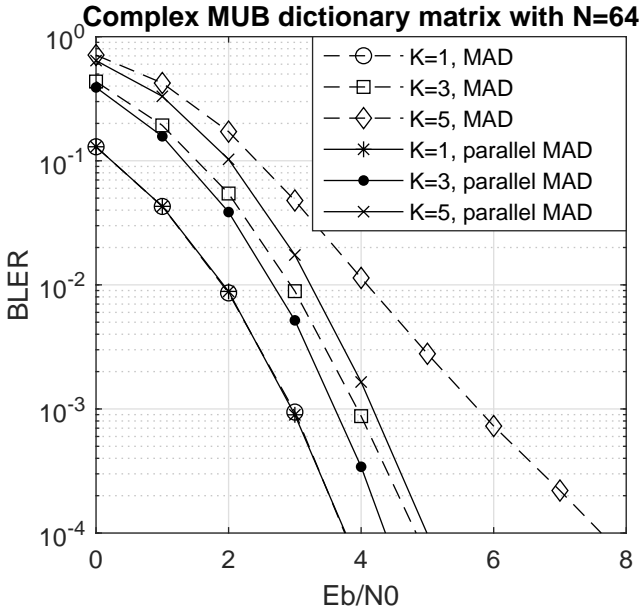}
   \caption{Comparing MAD with parallel MAD.}  
   \label{fig:MAD_vs_pMAD}
\end{figure}


We consider introducing \emph{random} phases to the columns of the complex MUB dictionary matrix based on the following reasoning. 
With $i$ being an index of one of the active columns from the sparse signal support set $\mc{S}$, consider
the inner product $\langle \vec{y},\vec{a}_i \rangle = \beta_i + \sum_{k \in \mc{S}, k \neq i} \beta_k \langle \vec{a}_k,\vec{a}_i \rangle  + \langle \vec{v},\vec{a}_i\rangle $. MAD decoder is prone to error when the net interference from other active columns has high magnitude. For the complex MUB dictionary matrix with $N=64$, the inner product between any two non-orthogonal columns belong to 
the set $\{\frac{+1}{\sqrt{N}},\frac{-1}{\sqrt{N}},\frac{+i}{\sqrt{N}},\frac{-i}{\sqrt{N}}\}$. 
Due to this property, there are many possible support sets $\mc{S}$, for which the interference terms can add coherently to result in a high magnitude. To overcome this problem, we introduce a random phase to each column of the dictionary matrix to minimize the constructive addition of interfering terms. Specifically, the random phase 
dictionary matrix is obtained as $\vec{A}\vec{\Phi}$ where $\vec{A}$ is the original (zero phase) 
MUB dictionary matrix and $\vec{\Phi}$ is a diagonal matrix with random diagonal entries $\exp\{j\theta_i\}$, with $\theta_i$'s  being i.i.d.  uniform over $[0,2\pi]$.  
In Fig. \ref{fig:Phase}, the comparison shows that random phase matrix has gains over zero phase matrix, when the sparsity leve $K$ is close to $\frac{1}{2\mu}$. 
In fact, in Figures \ref{fig:ex_vs_sb} and \ref{fig:MAD_vs_pMAD}, the plots corresponding to $K=5$ involve random phase MUB matrix.



\begin{figure}
    \centering
    \includegraphics{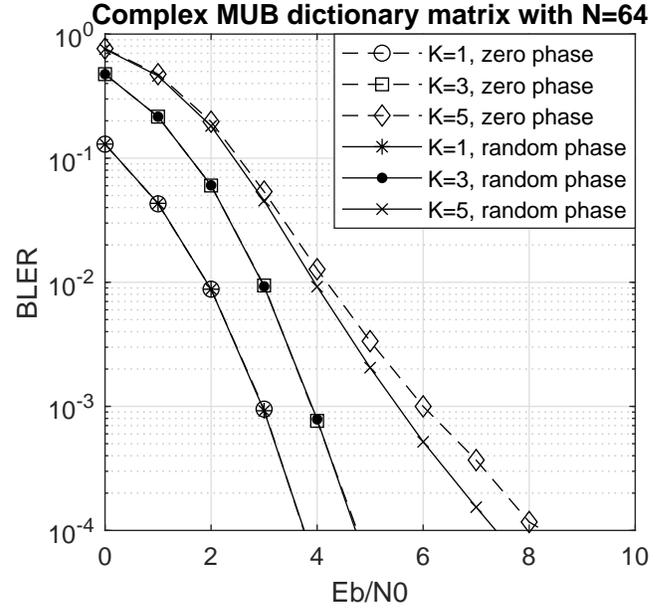}
   \caption{Impact of Random Phase.}
    \label{fig:Phase}
\end{figure}

\subsection{Comparison with Existing Error Control Codes} 
In this subsection, we study the performance of SSC scheme with parallel MAD decoding for various code rates and also compare with some 
of the existing error control codes.  In the conventional error control coding terminology, an $(n,k)$ coding scheme 
takes $k$ information bits and maps it into a (real) codeword of length $n$, with the ratio $\frac{k}{n}$ being referred as the code rate of the scheme. Hence, the 
SSC scheme can be compared directly with conventional error control coding schemes with identical code rates (measured in bits per real dimension). To match with the
conventional notation, we denote the SSC scheme conveying $N_b$ bits 
using real dictionary matrix with column length $N$ as an $(N, N_b)$ coding scheme, while with complex dictionary matrix,
we have $(2N,N_b)$ scheme. 

\begin{figure}
    \centering
    \includegraphics[scale=1]{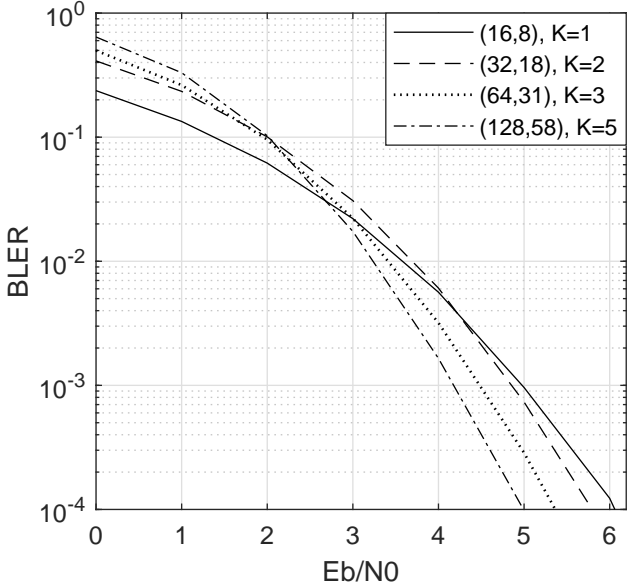}
   \caption{Complex MUB dictionary with code rate $\approx \frac{1}{2}$.}
    \label{fig:ratehalf}
\end{figure}


\begin{figure}
    \centering
    \includegraphics{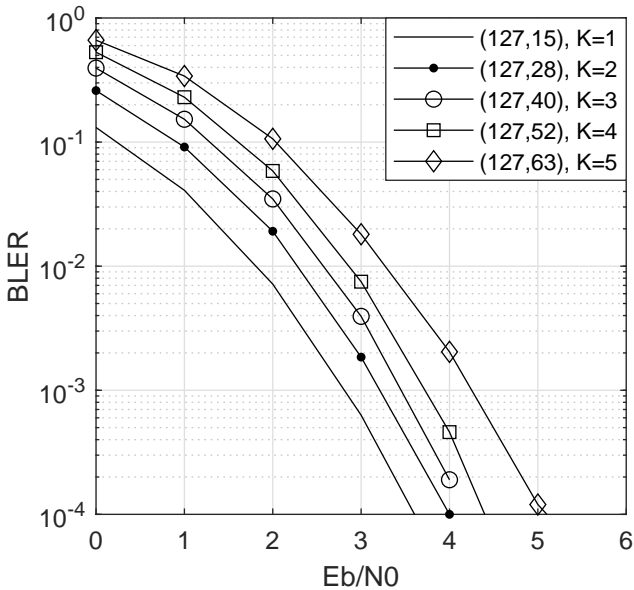}
    \caption{Gold Code Dictionary matrix with length $N=127$.}
    \label{fig:127gold}
\end{figure}

\begin{figure}
    \centering
    \includegraphics[scale=1]{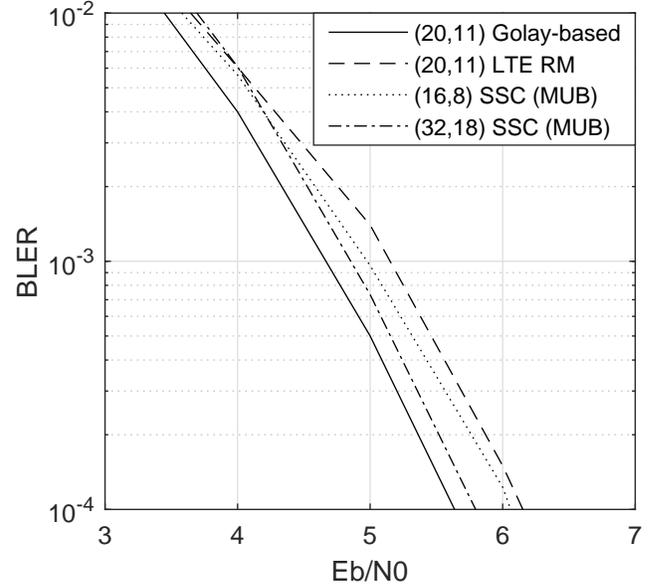}
   \caption{Comparsion with very small block length codes.}
    \label{fig:len20}
\end{figure}

\begin{figure}
    \centering
    \includegraphics[scale=1]{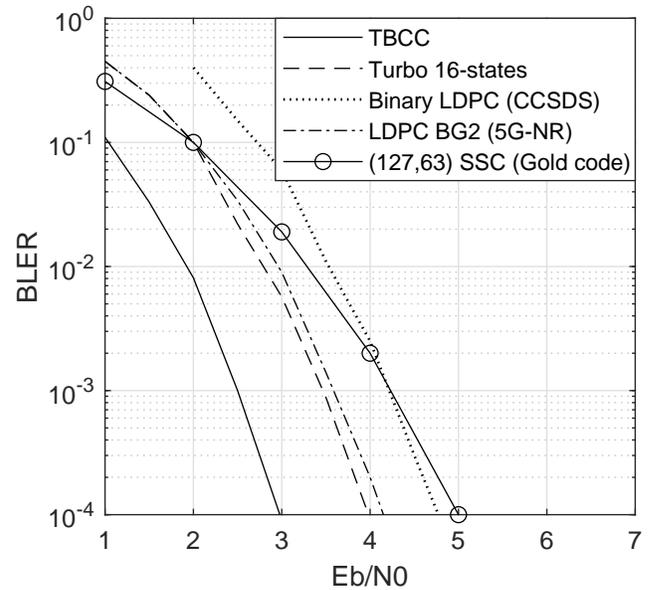}
   \caption{Comparison with exisiting $(128,64)$ Codes.}
    \label{fig:len128}
\end{figure}

In \figref{ratehalf}, we compare the performance of complex MUB dictionary matrices with various values of $N$ and $K$, while setting the code rate (bits 
per real dimension) approximately equal to half. We see that, for low $E_b/N_0$ values, smaller $N$ performs better. 
On the other hand, when small BLER is desired, larger block lengths are better. Here random phase is introduced when $K\geq \frac{1}{2\mu}$. 
In \figref{127gold}, we study the performance of Gold code dictionary matrix with $N=127$ for various values of code rates, by
varying the sparsity level $K$.

We present performance comparison with some of the existing error control coding schemes 
 with code rates close to half, for very small block lengths in \figref{len20}, 
and moderate block lengths in \figref{len128}. 
In \figref{len20}, the BLER performance of $(20,11)$ Golay based code and $(20,11)$ Reed Muller (RM) code used in the LTE standard are
taken from the plots given in \cite{van2018short}. We compare these codes with our $(16,8)$ and $(32,18)$ SSC schemes 
(complex MUB dictionary matrix with $K=1$ and $K=2$ respectively). 
We find that our schemes perform better than the RM code used in the LTE standard. In \figref{len128}, using the plots given in \cite{cocskun2019efficient}, we compare the performance of our $(127,63)$ SSC scheme  (Gold code dictionary matrix with $K=5$) with some of the existing $(128,64)$ error control codes: tail biting convolutional code (TBCC) with constraint length 14, binary LDPC codes used 
in the CCSDS standard,  LDPC codes (base graph 2) from 5G-NR standard and Turbo code with 16 states. More details about these existing codes are given in \cite{cocskun2019efficient}. Our proposed scheme performs comparable to the binary LDPC code from the CCSDS standard.

\section{Application in Multi-User Scenarios} \label{sec:mus}
In this Section, we discuss on how the SSC encoding and MAD decoding scheme can be used in multi-user 
scenarios. First, we note that linear codes are naturally unsuited for multi-user scenarios. To illustrate this, consider 
the simple case of two users who employ the same linear code. If $\vec{s}_1$ is the codeword sent by the user-1 and $\vec{s}_2$
is the codeword sent by the user-2, then their algebraic sum $\vec{s}_1+\vec{s}_2$ is another valid codeword for both users. 
This makes it impossible to recover the individual codewords from the superposition. In order to overcome this problem, power-domain non-orthogonal
multiple access (NOMA) techniques in the literature \cite{saito2013non} focus on grouping the users such that user-2 has much smaller channel gain than user-1. 
In this case, the superposition can be represented as $\vec{s}_1 + h \vec{s}_2$ with $|h| \ll 1$, facilitating successive interference cancellation based decoding 
\cite{saito2013non}. With this grouping, the data rate for user-2 will be very small compared to that of user-1. On the other hand, the SSC encoding technique in \secref{ssc} provides a straightforward way to communicate in multi-user scenarios.

\subsection{Encoding for Multiple Users}
Using the SSC scheme with sparsity level $K$ described in \secref{ssc}, we can support $P$-user multiple access channel, or $P$-user broadcast channel 
or $P$-user interference channel \cite{cover1999elements,el2011network}, for any $P \leq K$. First, we illustrate how the SSC scheme can be employed to generate the codeword of each user based on the user's information bits. 
As before, we partition the dictionary matrix $\vec{A}$ into $K$ subblocks, with subblock $\vec{A}_k$ 
having $L_k$ number of columns. These $K$ subblocks are divided among $P$ users, with $\mc{A}_i = \{\vec{A}_{i,1},\cdots,\vec{A}_{i,K_i}\} \subset 
\{\vec{A}_1,\cdots,\vec{A}_K\}$ denoting the ordered set of $K_i$ subblocks assigned to user-$i$. 
Note that $\mc{A}_i \cap \mc{A}_j = \emptyset$ if $i \neq j$ and $\displaystyle \sum_{i=1}^K K_i = K$. The codeword for user-$i$ is obtained as
\begin{equation}
\vec{s}_i = \sum_{k=1}^{K_i} \beta_{i,k} \vec{a}_{i,k} \label{eq:cwi}
\end{equation}
where symbols $\{\beta_{i,1},\cdots,\beta_{i,K_i}\}$ are chosen from $M_i$-ary constellation and the column $\vec{a}_{i,k}$ is chosen from the subblock $\vec{A}_{i,k}$
for $1 \leq k \leq K_i$. Denoting the number of columns in $\vec{A}_{i,k}$ as $L_{i,k}$, the total number of bits that can be conveyed for user-$i$ is 
\begin{equation}
N_{b_i} = K_i \floor{ \log M_i } + \sum_{i=1}^{K_i} \floor{\log L_{i,k}}. \label{eq:nbi} 
\end{equation}
Now, we will see how these codewords can be used in MAC, BC and IC.

\subsubsection{Multiple Access Channel}
Multiple access channel is equivalent to an uplink scenario in a cellular network, where $P$ users are sending their data to a single receiver. 
In the MAC, the encoding is done independently by each user, which coincides with the SSC based procedure in \eqref{cwi}.
The observation at the receiver is 
\begin{eqnarray}
\vec{y} &=& \sum_{i=1}^P \vec{s}_i + \vec{v},  \\
&=& \sum_{i=1}^P \sum_{k=1}^{K_i} \beta_{i,k} \vec{a}_{i,k} + \vec{v}. \label{eq:mac}
\end{eqnarray}
The decoding is done jointly at the receiver, which can be done using the MAD algorithm (and parallel MAD), 
which recovers the support of the active columns $\{\vec{a}_{i,k}\}$ and the corresponding modulation symbols $\beta_{i,k}$ for each user, 
using the received signal $\vec{y}$ in \eqref{mac}. 

If all the users employ the same constellation with $M_i=M, \forall i$, then the superposition of
codewords from $P$ users in \eqref{mac} will have identical structure to the codeword generated based on the SSC encoding procedure 
(with sparsity level $K=\sum_iK_i$) from \secref{ssc}. 
As a result, the total number of bits of all the users $\sum_{i=1}^P N_{b_i}$ (from \eqref{nbi}) will be equal to $N_b$ from \eqref{nb}. In addition,
the overall BLER performance of this MAC channel (probability that all the bits of all the users are decoded correctly with respect to the overall energy spent per bit) will coincide with the performance of the corresponding single user case. 
Note that, the performance of this single user case in AWGN has already been studied in \secref{sims}.

Now, considering the case where user-$i$ has channel gain $h_i$, the received signal is
\begin{equation}
\vec{y} = \sum_{i=1}^P h_i g_i \vec{s}_i + \vec{v}, \label{eq:mac1}
\end{equation}
where $g_i$ denotes the transmit gain (power control) employed by user-$i$. If the gains $g_i$ are chosen such that $h_i g_i = 1, \forall i$, then 
the AWGN performance of the above MAC model \eqref{mac1} will coincide with AWGN performance of the corresponding single user case. This implies that, if we group the 
users such that $|h_i|$ are (approximately) equal, then power control gains $|g_i|$ can be (approximately) same. Hence, in our SSC based scheme for MAC, grouping users with similar channel gains and dividing the dictionary matrix into (approximately) equal sizes among the users is beneficial. 
This is in contrast with the power-domain NOMA techniques in the literature \cite{saito2013non}, where a high channel gain user is typically grouped with low channel gain user. 
If different users have different channel gains, different power constraints and different rate requirements, then there are open issues in the SSC based scheme, such 
as segmenting the dictionary matrix among users, choosing the constellation size for each user and the modifications required for the MAD decoding. 

\subsubsection{Broadcast Channel}
Broadcast channel is similar to the downlink scenario in a cellular network, where the base station transmits respective information messages to $P$ users. Here, 
encoding is done jointly at the base station and the decoding is done by each user separately. It is well known that for degraded broadcast channel (for instance, Gaussian broadcast channel), superposition coding is optimal \cite{cover1999elements,el2011network}. Once the codebooks of all the users are designed (jointly), 
superposition coding simply chooses the codeword for $p^{th}$ user from his codebook using his information bits (and independent of the other users' codewords) and 
then transmits the \emph{sum} of the codewords of all the users. The SSC encoding procedure in \eqref{cwi} can mimic the superposition coding, in a straightforward way. 
Once the dictionary matrix is chosen, subblocks are segmented and allotted among the users, the codeword for each user can be obtained based on his own 
information bits as given by \eqref{cwi}, and the transmitter sends the sum of all the users' codewords as
\begin{equation}
\vec{s} = \sum_{i=1}^P \vec{s}_i. \label{eq:spc}
\end{equation}
Note that the above sum \eqref{spc} resembles the codeword generation of single user scenario \secref{ssc}. Received signal at the user-$i$ is given by
\begin{equation}
\vec{y}_i = \vec{s} + \vec{n_i},
\end{equation}
where $\vec{n}_i$ is the noise at the user-$i$, with variance $\sigma_i^2$. MAD decoding from \secref{ssc} can be employed by each user, which recovers the 
active columns present in $\vec{s}$ and the corresponding modulation symbols. Hence, in this approach, each user recovers the information sent to other users in addition
to his own information. The user with the highest noise variance ($\argmax_i \sigma_i^2$) will have the worst error performance (assuming the noise distributions are 
same except for the variance). On the other hand, if all the users have 
same noise variance, the performance (the probability that all the users received all their bits correctly) will coincide with the corresponding single user scenario 
(with the same noise variance). If the users' channel quality is asymmetric, and there is private information for good channel quality users (which the low channel
quality users should not be able to decode), then developing sparse coding based techniques is an open problem. 

\subsubsection{Interference Channel}
In the interference channel, there are $P$ transmitters and $P$ receivers. Each transmitter sends information to a corresponding intended receiver. With $i^{th}$ 
transmitter generating codeword as in \eqref{cwi}, the received signal at the $i^{th}$ receiver is given by
\begin{equation}
\vec{y}_i = \vec{s}_i + \sum_{\stackrel{j=1,\cdots,P}{j \neq i}} h_{i,j} \vec{s}_j + \vec{n}_i, \label{eq:ic}
\end{equation}
where $h_{i,j}$ denotes the channel gain from $j^{th}$ transmitter to the $i^{th}$ receiver.  Without loss of generality, we have taken $h_{i,i} = 1$. MAD decoding
employed at the $i^{th}$ receiver recovers the codewords of all the transmitters. Again, if $|h_{i,j}| = 1, \forall i,j$, and the noise statistics are identical across 
all the receivers, then the decoding performance (successful recovery of all the codewords) of all the receivers will coincide 
with the corresponding single user case. In the \emph{strong interference} 
$|h_{i,j}| \gg 1$ regime \cite{cover1999elements,el2011network}, MAD decoding can be modified to first decode the messages of 
all the interfering users (by restricting the correlations to the subblocks of 
interfering users), cancel the interference, and then proceed to find the active columns in the subblocks of the intended user. In the \emph{weak interference} 
$|h_{i,j}| \gg 1$ regime \cite{cover1999elements,el2011network}, MAD decoder can first find the message of the intended user 
directly by restricting the correlations to the subblocks of the intended user. Detailed study of  the 
SSC-MAD based techniques in the strong and weak interference regimes  can be explored in a future work.

Let us illustrate the gains of using our SSC schemes in multi-user channels when compared to using conventional error control codes in orthogonal multiple access fashion. Using the $(20,11)$ Golay based code (which gives the best performance among very short length codes from \figref{len20}) in orthogonal multiple access, 
each user gets  a code rate of $11/20$ and achieves individual BLER of $10^{-4}$ at $E_b/N_0$ of around $5.5$ dB. Using our $(127,63)$ SSC scheme (Gold code dictionary matrix with $K=5$) shown in \figref{len128}, we can transmit a total of $63$ bits using $127$ real dimensions to (up to) 5 different users in MAC or BC or IC. 
If all the receivers in the multi-user channel have the same AWGN variance, our scheme will get an \emph{overall} (accounting all the users together) code rate of $63/127$  
and achieve overall BLER of $10^{-4}$ at overall $E_b/N_0$ of $5$ dB. Hence, our scheme used in multi-user channels gives $0.5$ dB gain over one of the best codes for very short block lengths.

\section{Concluding Remarks} \label{sec:con}
In this Section, we present the limitations of the proposed encoding techniques, various directions for future work and give final remarks on our work.

\subsection{Limitations of the Sparse Coding Techniques}

1) The SSC-MAD scheme can have non-zero probability of error, even in the absence of noise. This is because, each active column \emph{interferes} with 
the other active columns, causing the MAD decoder to select an inactive column. Even when $K < \frac{1}{2}(\mu^{-1} + 1)$, MAD decoder can make error in decoding 
the constellation symbols $\{\beta_i\}$ for $M>1$, which in turn causes error in recovering the support set $\mc{S}$.


2) SSC scheme can have a large alphabet size. 
When the block size $N$ is a power of $2$, the  complex MUB matrices have alphabet size of $4$, as noted in \secref{dic}. With sparsity level $K$ encoding and $M$-ary
constellation symbols, the alphabet size of SSC codeword \eqref{cw1} can be up to $4KM$. On the other hand, the binary linear codes have alphabet size of $2$. 
However, alphabet size may not be a major issue in OFDM based communication systems. Even when the input is binary, 
the output of the FFT block in OFDM transmitter resembles Gaussian symbols (for large FFT sizes) \cite{jiang2008overview}.  

3) From \lemref{rec}, we infer that the sparsity level should be of the order of  $\frac{1}{\mu}$ for good recovery performance. With block length $N$, the mutual coherence of MUB and Gold code dictionary matrices are of the order of $\frac{1}{\sqrt{N}}$. 
With the sparsity level $K = \gamma \sqrt{N}$ for some fixed $\gamma$, the bits per dimension 
of the SSC and SC schemes ($N_b/N$) converge to $0$ as $N \rightarrow \infty$. Hence, our sparse coding based schemes are not suitable 
for large block lengths with high code rates.

\subsection{Directions for Future Work}

1) MAD decoding involves computation of inner products of the observation vector with all the columns of dictionary matrix. Efficient ways to compute these 
inner products for MUB and Gold code dictionary matrices can be explored in future.

2) Sufficient conditions for successful support recovery for OMP algorithm in the presence of noise, exist in the literature \cite{cai2011orthogonal}. Some of these conditions 
can be modified to give bounds on the error performance of MAD algorithm, especially, when the constellation size $M=1$. On the other hand, obtaining tight
bounds on the probability of block error for MAD decoder in Gaussian noise is a challenging problem, especially, for the constellation size $M>1$.

3) The SSC scheme has a penalty in the total number of bits $N_b$  transmitted in a block \eqref{nb} when compared to the sparse coding scheme \eqref{nb1},
for $K>1$. Other simple and efficient ways to map the information bits to the support set which give higher number of bits than SSC scheme can be developed in  a future work.

4) MAD algorithm is inspired by the matching pursuit based sparse signal recovery algorithms. There are also convex programming based algorithms for sparse signal 
recovery, such as $\ell_1$ norm minimization \cite{chen2001atomic} and dantzig selector \cite{candes2007dantzig}. To solve these convex programs, there are several iterative techniques available in the literature, such as ISTA \cite{daubechies2004iterative}, FISTA \cite{beck2009fast} and AMP \cite{maleki2011approximate}. We can study the decoding performance of the 
SSC scheme with these other sparse signal recovery algorithms. In addition, using the deep unfolding principle, 
some of the iterative sparse signal recovery algorithms such as AMP, have been implemented using
deep learning architectures \cite{metzler2017learned}. Such deep learning architectures can be employed at the receiver to decode the SSC scheme.

5) As discussed in \secref{mus}, there are several avenues for further investigation of sparse coding based techniques in multi-user scenarios, especially, when 
there is asymmetry in the channel conditions, power constraints and rate requirements among users.

6) We can explore other options for the dictionary matrix construction. 
There are constructions available in the literature for real MUB matrices under some conditions 
on the block length $N$ \cite{cameron1991quadratic},
which give a real (binary) dictionary matrix with $\frac{N^2}{2}$ columns and mutual coherence $\frac{1}{\sqrt{N}}$. 
Real MUB occupies only half the number of real dimensions for the same block length $N$ when compared to the complex MUB. So real MUB will achieve the same code rate (bpd) using a smaller sparsity level $K$ and hence its BLER performance will be better than the corresponding complex MUB.  There are also approximate MUB \cite{amub} constructions, which can give dictionary matrix with more than $N^2$ columns 
at the expense of  mutual coherence being higher than $\frac{1}{\sqrt{N}}$. Other quantum designs such as SIC-POVM, approximate SIC-POVM and Zadoff-Chu sequences from CDMA can also be used for obtaining the dictionary matrix.
 
7) Extensions of SSC encoder and MAD decoder for frequency selective channels and single/multi-user MIMO channels can be 
explored in future.

8) Applicability of sparse coding based techniques for the storage systems can be explored in a future work. 
Note that, with MUB dictionary matrix and sparsity level $K=1$, the alphabets of the codewords are from a (rotated) QPSK constellation 
and hence can be represented using 2 bits. However, the sparse coding alphabet size increases with $K$, making it unsuitable for storage
when $K$ is large. 

9) Communications with sparse signal based encoding and matching/correlation based decoding is very simple and intuitive. 
Using this sparse signal based communication to model/understand some of the naturally occurring communications, 
for instance, communications based on neurological signals, can be explored in a future work.

\subsection{Final Remarks}
We proposed SSC encoding scheme and MAD decoding scheme for communications, inspired by sparse signal processing. Our proposed technique is a non-linear coding scheme, 
whose implementation is easy to understand. With MUB and Gold code dictionary matrices, the proposed scheme gives competitive performance when compared to some of the commonly used linear codes, for small block lengths and low code rates. Unlike linear codes, our proposed sparse coding based techniques extend neatly to multi-user scenarios. Our schemes can be straightforwardly used in applications where there are several users with small number of information bits to transmit-to and/or receive-from. Such applications can include communicating control information to/from several users in a cellular network or in a vehicular communication system or communications in internet of things (IoT) applications.

\section*{Acknowledgment}
We would like to thank our colleague Pradeep Sarvepalli for introducing us to the delightful world of mutually unbiased bases.

\bibliographystyle{ieeetr}
\bibliography{ref}

\end{document}